\documentclass[runningheads,11pt]{article}
\usepackage[affil-it]{authblk}
\usepackage{amssymb,amsmath, amsthm}
\usepackage[vlined, ruled, linesnumbered]{algorithm2e}
\usepackage{tikz}
\usepackage{xspace}
\usepackage{booktabs}
\newtheorem{thm}{Theorem}
\newtheorem{cor}[thm]{Corollary}
\newtheorem{lem}[thm]{Lemma}
\newtheorem{prop}[thm]{Proposition}
\newtheorem{obser}[thm]{Observation}

\newtheorem{definition}{Definition}
\usepackage{todonotes}
\usepackage{tabularx}
\usepackage{MnSymbol}

\usepackage[T1]{fontenc}
\usepackage[utf8]{inputenc}

\newcommand{\dist}{\mathop{\mathrm{dist}}}

\def\O{\mathcal{O}}
\def\I{\mathcal{I}}
\def\J{\mathcal{J}}
\newcommand{\vc}{{\sc Vertex Cover}\xspace}
\newcommand{\fix}{{\sc Fix}\xspace}
\newcommand{\listfix}{{\sc List-Fix}\xspace}
\newcommand{\rfix}[1]{{\sc $#1$-Fix}\xspace}
\newcommand{\rcoloring}[1]{{\sc $#1$-Coloring}\xspace}
\newcommand{\preext}[1]{{\sc $#1$-PreExt}\xspace}
\newcommand{\fixnumb}[3]{{\overline{\chi}^{#2}_{#1}(#3)}\xspace}
\newcommand{\fixnumber}[1]{{\overline{\chi}(#1)}\xspace}
\newcommand{\msi}{{\sc Multicolored Subgraph Isomorphism}\xspace}
\newtheorem*{eth-env}{Exponential Time Hypothesis}
\newtheorem*{seth-env}{Strong Exponential Time Hypothesis}

\usetikzlibrary{shapes,snakes}

\DeclareMathOperator*{\tw}{tw}

\title{Fixing improper colorings of graphs\footnote{An extended abstract of this paper was presented on the conference SOFSEM 2015 \cite{sofsem}.}}

\author{Valentin Garnero}
\affil{Université d'Orl\'eans, INSA Centre Val de Loire\\ 
LIFO, 45067 Orl\'eans, France.\\
E-mail: \{valentin.garnero, mathieu.liedloff, pedro.montealegre\}@univ-orleans.fr}
\author{Konstanty Junosza-Szaniawski}
\affil{Warsaw University of Technology, Faculty of Mathematics and Information Science,
Koszykowa 75, 00-662 Warszawa, Poland.\\
E-mail: \{k.szaniawski, p.rzazewski\}@mini.pw.edu.pl}
\author[1]{Mathieu Liedloff}
\author[1]{Pedro Montealegre}
\author[2,3]{Pawe{\l} Rz\k{a}\.zewski\thanks{Supported by ERC Starting Grant PARAMTIGHT (No. 280152).}}
\affil{Institute of Computer Science and Control, Hungarian Academy of Sciences (MTA SZTAKI), Budapest, Hungary}

\begin{document}
\maketitle

\begin{abstract}
In this paper we consider a variation of a recoloring problem, called the Color-Fixing. Let us have some non-proper $r$-coloring $\varphi$ of a graph $G$.
We investigate the problem of finding a proper $r$-coloring of $G$, which is ``the most similar'' to $\varphi$, i.e., the number $k$ of vertices that have to be recolored is minimum possible. We observe that the problem is NP-complete for any fixed $r \geq 3$, even for bipartite planar graphs. Moreover, it is $W[1]$-hard even for bipartite graphs, when parameterized by the number $k$ of allowed recoloring transformations. On the other hand, the problem is fixed-parameter tractable, when parameterized by $k$ and the number $r$ of colors.

We provide a $2^n \cdot n^{\mathcal{O}(1)}$ algorithm for the problem and a linear algorithm for graphs with bounded treewidth. We also show several lower complexity bounds, using standard complexity assumptions.
Finally, we investigate the {\em fixing number} of a graph $G$. It is the minimum $k$ such that $k$ recoloring transformations are sufficient to transform \emph{any} coloring of $G$ into a proper one.
\end{abstract}

\section{Introduction}
Many problems in real-life applications have a dynamic nature. When the constraints change, the previously found solution may no longer be  optimal or even feasible. Therefore often there is a need to recompute the solution (preferably using the old one). This variant is called a {\em reoptimization} and has been studied for many combinatorial problems, e.g. TSP (see Ausiello {\em et al.} \cite{AEMP}), Shortest Common Superstring (see Bil\`{o} {\em et al.} \cite{BBKKMSZ}) or Minimum Steiner Tree (see Zych and Bil\`{o} \cite{ZB}). We also refer the reader to the paper of Shachnai {\em et al.} \cite{STT}, where the authors describe a general model for combinatorial reoptimization.

Another family of problems, in which we deal with transforming one solution to another, is {\em reconfiguration}. Here we are given two feasible solutions and want to transform one into another by a series of simple transformations in such a way that every intermediate solution is feasible (see e.g. Ito {\em et al.} \cite{IDHPSUU}). When we consider a reconfiguration version of the graph coloring problem, we want to transform one proper coloring into another one in such a way that at every step we can recolor just one vertex and the coloring obtained after this change is still proper.
Bonsma and Cerceda showed that deciding if a proper $k$-coloring $\varphi$ can be transformed into another $k$-coloring $\varphi'$ is PSPACE-complete for every $r \geq 4$~\cite{BC}.

A special attention has been paid to determining if a given graph $G$ is {\em $r$-mixing}, i.e., if for any two proper $r$-colorings of $G$ you can transform one into another (maintaining a proper $r$-coloring at each step). Cereceda {\em et al.} \cite{CHJ,CHJ2,CHJ3} characterize graphs, which are 3-mixing, and they provide a polynomial algorithm for recognizing them. 
There are also some results showing that a graph $G$ is $f(G)$-mixing, where $f(G)$ is some invariant of $G$. For example, Jerrum \cite{Jerrum} showed that every graph $G$ is $(\Delta(G)+2)$-mixing. This bound was later refined by Bonamy and Bousquet \cite{BB}, who proved that every graph is $(\chi_g(G)+1)$-mixing, where $\chi_g(G)$ denotes the Grundy number of $G$, i.e., the highest possible number of colors used by a greedy coloring of $G$. Clearly $\chi_g(G) \leq \Delta(G)+1$.

Another direction of research in $r$-mixing graphs is the maximum number of transformations necessary to obtain one $r$-coloring from another one, i.e., the distance between those colorings. Bonamy and Bousquet \cite{BB} show that if $r \geq \tw(G)+2$ (where $\tw(G)$ denotes the {\em treewidth} of $G$), then any two $r$-colorings of $G$ are at distance of at most $2(n^2+n)$, while for $r \geq \chi_g(G)+1$, any two $r$-colorings are at distance of at most $4 \cdot \chi_g(G) \cdot n$.

A slightly different problem has been considered by Felsner {\em et al.} \cite{FHS}. They also transformed one $r$-coloring to another one using  some local changes, but did not require the initial coloring to be proper (the final one still has to be proper). Also, a vertex could be recolored to color $x$ if it did not have any neighbor colored with $x$ (strictly speaking, any out-neighbor, as the authors were considering directed graphs). They showed that if $G$ is a 2-orientation (i.e., every out-degree is equal to 2) of some maximal bipartite planar graph (i.e., a plane quadrangulation), then every proper 3-coloring of $G$ could be reached in $\O(n^2)$ steps from any initial (even non-proper) 3-coloring of $G$. Similar results hold for 4-colorings and 3-orientations of maximal planar graphs (i.e., triangulations).

In this paper we consider a slightly different problem. We start with some (possibly non-proper) $r$-coloring and ask for the minimum number of transformations needed to obtain a proper $r$-coloring (any proper $r$-coloring, not a specific one). We are allowed to change colors of vertices arbitrarily, provided that we recolor just one vertex in each step. 
We mainly focus on the computational aspects of determining if, starting with some given $r$-coloring of $G$, we can reach a proper $r$-coloring in at most $k$ steps.

The paper is organized as follows. 
In Section \ref{sec:complexity} we show that our problem is NP-complete for any $r \geq 3$, even if the input graph is planar and bipartite (here $k$ is a part of the input).
In Section \ref{sec:exact} we provide an $2^n \cdot n^{\O(1)}$ algorithm for the problem and show that it is essentially optimal under the ETH.
In the next two Sections we focus on the parameterized complexity (we refer the reader to \cite{book,DF} for an introduction to the parameterized complexity theory). First, we present an algorithm solving the problem in time $(2(r-1))^k \cdot n^{\O(1)}$, which showns that our problem is FPT, when parameterized by $k+r$ (Section \ref{sec:parameterized}).
Then we show that the problem is $W[1]$-hard, when parameterized by $k$ only, wven if the input graph is bipartite. Moreover, we present an almost tight lower bound, excluding an algorithm with running time $f(k) \cdot r^{o(k / \log k)} \cdot n^{\O(1)}$ for any function $f$, under the ETH.
Finally, we show that for any $r \geq 3$, the  problem does not admit a kernel parameterized in $k$ (unless NP $\subseteq$ coNP / poly), even in the input graph is bipartite.

In Section \ref{sec:treewidth} we provide an algorithm solving the problem for graphs with bounded treewidth  and show that is it essentially optimal, under the SETH.
In Tables \ref{summary-complexities} and \ref{summary-algorithms} you can find a summary of the most important results of the paper.

The last section of the paper, Section \ref{sec:number}, is purely combinatorial.  We investigate the {\em fixing number} of $G$, i.e., the maximum (over all initial colorings $\varphi$) distance  from $\varphi$ to a proper coloring of $G$. We provide some combinatorial bounds and suggest directions for future research.

\begin{table}[h!]
\centering
\begin{tabular}{ccccc}
\toprule
  $r \leq 2$ & $r \geq 3$ & $k$ & $k+r$ &  $\tw + r$\\
\midrule
P (Prop~\ref{prop:easy})& NP-c (Th~\ref{thm:npc}) &  $W[1]$-h  (Th~\ref{thm:w1hard}) & FPT  (Th~\ref{fpt})  & FPT  (Th~\ref{thm:tw})\\
\bottomrule
\end{tabular}
\vspace{0.2cm}
\caption{The summary of parameterized complexity results for the \rfix{r} problem and different parameters.}
\label{summary-complexities}
\end{table}

\begin{table}[h!]
\centering
\begin{tabular}{l | ccc}
\toprule
 & $n$ & $k+r$ &  $\tw + r$\\
\midrule
upper bound & $2^n$ (Cor~\ref{cor:exact-algo})&  $(2(r-1))^k$ (Th~\ref{fpt-algo})& $r^t$ (Th~\ref{thm:tw})\\
lower bound & $2^{o(n)}$ (Cor~\ref{cor:exact-eth}) &  $r^{o(k / \log k)}$ (Cor~\ref{cor:almost-optimal-k})& $(r-\epsilon)^t$ (Cor~\ref{cor:tw-seth})\\
\bottomrule
\end{tabular}
\vspace{0.2cm}
\caption{The upper and lower complexity bounds for the problem under different parameterizations. We suppress polynomial factors. Lower bounds should be read: ``there is no algorithm with this complexity, unless the ETH/SETH fails''.}
\label{summary-algorithms}
\end{table}

\section{Preliminaries}
For a natural number $r$, by $[r]$ we denote the set $\{1,2,..,r\}$.
By an {\em $r$-coloring} of a graph $G$ we mean any assignment of natural numbers from $[r]$ (called {\em colors}) to vertices of $G$. A coloring is {\em proper} if no two adjacent vertices get the same color. Note that there may be some colors that are not assigned to any vertex.

\subsection{Considered problems}

For two $r$-colorings $\varphi,\varphi'$, let $\varphi \ominus \varphi'$ denote the set $\{v \in V \colon \varphi(v) \neq \varphi'(v)\}$.
We also define the {\em distance} $\dist(\varphi,\varphi')$ between two $r$-colorings $\varphi,\varphi'$ to be their Hamming distance, i.e. $|\varphi \ominus \varphi'|$.

The problem we consider in this paper is formally defined as follows.

\vskip 5pt
\noindent\begin{tabularx}{\textwidth}{| X |}
\hline
{\bf Problem: Color-Fixing ({\sc Fix})}\\
{\bf Instance:} A graph $G$, integer $k$, integer $r$, an $r$-coloring $\varphi$ of $V(G)$.\\
{\bf Question:} Does there exist a proper $r$-coloring $\varphi'$ of $G$ such that $\dist(\varphi,\varphi') \leq k$?\\
\hline
\end{tabularx}
\vskip 5pt

If $r$ is a fixed integer, we have the following version of the problem.

\vskip 5pt
\noindent\begin{tabularx}{\textwidth}{| X |}
\hline
{\bf Problem: $r$-Color-Fixing (\rfix{r})}\\
{\bf Instance:} A graph $G$, integer $k$, an $r$-coloring $\varphi$ of $V(G)$.\\
{\bf Question:} Does there exist a proper $r$-coloring $\varphi'$ of $G$ such that $\dist(\varphi,\varphi') \leq k$?\\
\hline
\end{tabularx}
\vskip 5pt

Such a coloring $\varphi'$ is called a {\em witness} of an instance $\I$ of \fix~(\rfix{r}, resp.).
Obviously, if $r < \chi(G)$ (by $\chi(G)$ we denote the {\em chromatic number} of $G$, i.e., the smallest number of colors needed to color $G$ properly), then the answer is always {\sc No}.
By a {\em recoloring} of a vertex $v$ we mean an operation of changing the color assigned to $v$, obtaining another coloring $\varphi'$ such that $\varphi \ominus \varphi' = \{v\}$.

In the optimization version of the problem we ask for the minimum number $k$ of recolorings needed to transform $\varphi$ into a proper coloring of $G$.
Let $\fixnumb{\varphi}{r}{G}$ denote this minimum possible value of $k$. If $r < \chi(G)$, we define $\fixnumb{\varphi}{r}{G} := \infty$ for every $r$-coloring $\varphi$ of $G$.

\subsection{Conflict graph}

For an improper coloring $\varphi$ of $G$, let $G^{\varphi}$ denote a {\em conflict graph} of $G$ under the coloring $\varphi$, i.e., the subgraph of $G$ induced by the set of edges $\{uv \in E(G) \colon \varphi(u) = \varphi(v) \}$. Note that the conflict graph can be found in polynomial time. The simple observation below will prove useful.

\begin{obser}\label{conflict}
To make the coloring $\varphi$ of $G$ proper, we need to recolor at least one endvertex of each edge of $G^{\varphi}$.
\end{obser}

Thus we observe that there is a close relation between fixing an improper coloring and finding a vertex cover in the conflict graph. This relation is explicitly described in the following lemma.

\begin{lem} \label{vc-reduction}
Let $G=(V,E)$ be a graph on $n$ vertices and let $k \leq n$. 
Let $\varphi$ be the coloring of $G$ such that $\varphi(v) = k+1$ for all $v \in V$.
Then $(G,k,k+1,\varphi)$ is a {\sc Yes}-instance of the \fix~problem if and only if $G$ has a vertex cover of size at most $k$.
\end{lem}

\begin{proof}
We can assume that $G$ has no isolated vertices, since removing them does not change the size of the minimum vertex cover. Define $r = k+1$. Observe that $G^{\varphi} = G$.

First suppose that $(G,k,r,\varphi)$ is a {\sc Yes}-instance of \fix~with a witness $\varphi'$. Let $S = \varphi \ominus \varphi'$, clearly $|S| \leq k$.
By Observation \ref{conflict}, $S$ contains at least one vertex from each edge of $G^{\varphi} = G$. Thus $S$ is a vertex cover of $G$, of size at most $k$ and thus $(G,k)$ is a {\sc Yes}-instance of \vc.

Now suppose that $(G,k)$ is a {\sc Yes}-instance of \vc~and let $S = \{v_1,v_2\ldots,v_k\}$ the the vertex cover of size $k$ (note that we can always add some vertices to a smaller vertex cover to obtain a vertex cover of size exactly $k$). Define a coloring $\varphi'$ of $G$ in the following way:
$$
\varphi'(v) = 
\begin{cases}
i & \text{ if } v = v_i \in S\\
r & \text{ if } v \notin S.
\end{cases}$$
Note that $|\varphi \ominus \varphi'| = |S| = k$ and $\varphi'$ is a proper $r$-coloring of $G$. This shows that $(G,k,r,\varphi)$ is a {\sc Yes}-instance of \fix.
\end{proof}

\subsection{Computational assumptions}

When proving hardness results or lower bounds for algorithms, we often use some additional assumptions. The standard assumption for distinguishing easy and hard problem is P $\neq$ NP. However, these assumptions is too weak to give us any meaningful insights into the possible complexity of an algorithm solving an NP-hard problem.
Thus researchers use stronger assumptions to investigate hard problems in more detail. Such an assumption, typically used for this purpose, is the so-called {\em Exponential Time Hypothesis} (usually referred to as the ETH), formulated by Impagliazzo and Paturi \cite{ImpagliazzoPaturi}. We refer the reader to the survey by Lokshtanov and Marx for more information about the ETH and conditional lower bounds \cite{LokshtanovMS11}.
The version we present below (and is most commonly used) is not the original statement of this hypothesis, but its weaker version (see also Impagliazzo, Paturi, and Zane \cite{ImpagliazzoPZ01}).

\begin{eth-env}[Impagliazzo and Paturi \cite{ImpagliazzoPaturi}]
There is no algorithm solving every instance of 3-{\sc Sat} with $N$ variables and $M$ clauses in time $2^{o(N+M)}$.
\end{eth-env}

A stronger complexity assumption is the so-called {\em Strong Exponential Time Hypothesis}, also introduced by Impagliazzo and Paturi \cite{ImpagliazzoPaturi}. The version we present below is again the consequence of the original statement. It is worth mentioning that the SETH is indeed a stronger version of the ETH, as the SETH implies the ETH.

\begin{seth-env}[Impagliazzo and Paturi \cite{ImpagliazzoPaturi}]
For any $\epsilon>0$, there is no algorithm solving every instance of {\sc Cnf-Sat} with $N$ variables in time $(2-\epsilon)^{N} \cdot N^{\O(1)}$.
\end{seth-env}

{\emph Fixed-parameter tractability} is a central notion of parameterized complexity. We say that a computational problem is fixed-parameter tractable (FPT), when parameterized by some parameter $k$, if it can be solved in time $f(k) \cdot n^{\O(1)}$, where $n$ is the size of the input instance and $f$ is some function. On the other hand, a problem is in XP complexity class, if it can be solved in time $n^{f(k)}$, so for every fixed $k$ the complexity of polynomial, but the degree of the polynomial function depends on $k$. To distinguish these two classes one usually uses the notion of $W[1]$-hardness (we refer the reader to \cite{book} for a formal definition). Thus the next complexity assumption we are going to use in Theorem \ref{thm:w1hard} is that FPT $\neq W[1]$, which implies that a $W[1]$-hard problem cannot be solved in FPT time.

The last assumption we use (in the proof of Theorem \ref{no-poly-kernel}) is NP $\not\subseteq$ coNP / poly. It can be seen as a stronger variant of NP $\neq$ coNP, which in turn implies P $\neq$ NP. It is worth mentioning that NP $\subseteq$ coNP / poly would imply the collapse of the polynomial hierarchy to its third level.
This assumption is an important part of the framework used for showing the non-existence of polynomial kernels for parameterized problems, introduced by Bodlaender, Jansen, and Kratsch \cite{BJK}. The version of the framework presented below is a special case of the original one, yet it is strong enough for our purpose.

Let $\Pi$ be a graph problem, whose instance is a graph $G$. Let $\Pi^*$ be a parameterized problem, whose instance is $(G,k)$, where $G$ is a graph and $k$ is an integer (parameter). We say that $\Pi$ {\em cross-composes} into $\Pi^*$ if there exists an algorithm, which, given $t$ instances $G_1,G_2,\ldots,G_t$ of $\Pi$, works in time polynomial in $\sum_{i=1}^t |G_i|$, and produces an instance $(G^*,k^*)$ of $\Pi^*$ such that:
\begin{itemize}
\item $(G^*,k^*)$ is a {\sc Yes}-instance of $\Pi^*$ if and only if there exists $i \in [t]$ such that $G_i$ is a  {\sc Yes}-instance of $\Pi$,
\item $k^* \leq p(\max _{i=1}^t |V(G_i)| + \log t)$, where $p$ is some polynomial function.
\end{itemize}

Intuitively, we encode many (exactly $t$) instances of the problem $\Pi$ into one instance of $\Pi^*$ such that the size of the parameter $k^*$ is relatively small. Note that the size of $G^*$ can be huge. The theorem below shows that such a cross-composition can be used to refute the existence of a polynomial kernel for $\Pi^*$.

\begin{thm}[Bodlaender, Jansen, and Kratsch \cite{BJK}] \label{thm-no-kernel-framework}
Let $\Pi$ be an {\em NP}-hard problem and let $\Pi^*$ be a parameterized problem such that $\Pi$ cross-composes into $\Pi^*$. If $\Pi^*$ has a polynomial kernel, then {\em NP} $\subseteq$ {\em coNP / poly}.
\end{thm}

We refer the reader to the handbook \cite[Chapters 14 and 15]{book} for more information about complexity assumptions.

\section{Classical complexity results} \label{sec:complexity}

In this section we show that the \rfix{r} problem is NP-complete for all $r \geq 3$, even for restricted input graphs. Furthermore, we show how to adapt the known exact algorithm for computing partitions of graphs to solve \rfix{r}.

\subsection{Computational hardness of the problem}
First, observe that for $r=1$ and $r=2$ the problem is easy. 

\begin{prop}\label{prop:easy}
The \rfix{r} problem is polynomially solvable for $r \leq 2$.
\end{prop}
Indeed, if $r=1$, the problem clearly reduces to determining if the graph has no edges.
For $r=2$ the problem is also polynomial time solvable. If $G$ is not bipartite, the answer is {\sc No}. If $G$ is bipartite, the answer follows from the proposition below.

\begin{prop} \label{bipartite}
Let $G$ be a bipartite graph with bipartition classes $X$ and $Y$ and let $\varphi$ be a 2-coloring of $G$.
Then we have
$$\fixnumb{\varphi}{2}{G} = \sum_{ \substack{ C \colon \text{connected} \\ \text{component of }G}} 
\min \{ |\left ( X \ominus \varphi^{-1}(1) \right ) \cap V(C) |,| \left ( X \ominus \varphi^{-1}(2) \right ) \cap V(C)| \}.$$
\end{prop}

\begin{proof}
The claim follows from the observation that a connected bipartite graph has a unique 2-coloring.
Let $C$ be a connected component of $G$ and let $X',Y'$ denote its classes of bipartition. By $\varphi'$ we denote the restriction of $\varphi$ to $C$. To obtain a proper coloring of $C$, we either have to recolor the vertices from $X' \setminus \varphi'^{-1}(1)$ to color 1 and vertices from $Y' \setminus \varphi^{-1}(2)$ to color 2, or the other way around. Therefore the minimum number of recoloring operations needed to obtain a proper coloring of $C$ is equal to $\min \{ |X'  \ominus \varphi'^{-1}(1)|,|X' \ominus \varphi'^{-1}(2)| \}$. Clearly $X' \ominus \varphi'^{-1}(1) = \left ( X \cap V(C) \right ) \ominus \left( \varphi^{-1}(1) \cap V(C) \right) = \left(X \ominus \varphi^{-1}(1) \right) \cap V(C)$ (and symmetrically for $\varphi^{-1}(2)$). We repeat this for every connected component $C$ of $G$.
\end{proof}

Clearly, the \fix problem is in NP. Also, a graph $G$ with $n$ vertices is $r$-colorable if and only if $G$ can be recolored from any fixed coloring within at most $n$ steps.
Thus the \rfix{r} problem is NP-complete for any $r \geq 3$ (when the number $k$ of allowed recolorings is a part of the input). From this simple reduction, we obtain the following lemma, that will be useful later.

\begin{lem} \label{lem:equivalent}
If there exists an algorithm solving the \rfix{r} problem on a graph $G$ in time $f(G)$, then the \rcoloring{r} problem on $G$ can be solved in time $f(G) \cdot n^{\mathcal{O}(1)}$, where $n$ is the number of vertices of $G$.
\end{lem}

To derive slightly stronger complexity results, consider the problem \preext{r} defined as follows: 
\vskip 5pt
\noindent\begin{tabularx}{\textwidth}{| X |}
\hline
{\bf Problem: $r$-Precolor-Extension (\preext{r})}\\
{\bf Instance:} A graph $G$, a set $U \subseteq V(G)$ and a proper $r$-coloring $\varphi_U$ of $G[U]$.\\
{\bf Question:} Does there exist a proper $r$-coloring $\varphi$ of $G$ such that $\varphi(u) = \varphi_U(u)$ for all $u\in U$?\\
\hline
\end{tabularx}
\vskip 5pt

As shown by Kratochv\'il \cite{Kratochvil1993}, \preext{3} is NP-complete, even if the input graph is planar and bipartite.

Let $(G, U, \varphi_U)$ be an instance of \preext{3}, where $G$ has $n$ vertices. Let  $H$ and $\varphi$ be the graph and its coloring defined as follows. First, pick a copy of $G$ and color its vertices with color $1$. For each $v\in V(G)$, add $(r-3)$ groups of $(n+1)$ pending vertices,  where the $i$-th group is colored with the color $3+i$. Finally, for each node $u\in U$, add two more groups of $(n+1)$ pending vertices, each group colored with a different color from $\{1,2,3\} \setminus \{ \varphi_U(u)\}$. 

\begin{lem} \label{rpreext-reduction}
Fix $r \geq 3$. Let $(G, S, \varphi_U)$ be an instance of \preext{3}, and $H=H(G)$, $\varphi = \varphi(\varphi_U)$ be defined as above. Then $(G, U, \varphi_U)$  is a {\sc Yes}-instance of \preext{3} if and only if $(H, n, \varphi)$ is a {\sc Yes}-instance of \rfix{r}.
\end{lem}

\begin{proof}
First suppose that $(G,U,\varphi_U)$ is a {\sc Yes}-instance of \preext{3} with a witness $\varphi'$, i.e., $\varphi'$ is a proper $3$-coloring of $G$ (with colors  $\{1,2,3\}$) that extends $\varphi_U$. Let $\varphi''$ be the coloring of $H$ where the vertices in the copy of $G$ in $H$ are colored according to $\varphi'$, and the other vertices of $H$ are colored according to $\varphi$. Notice that $|\varphi'' \ominus \varphi| \leq n$, and $\varphi''$ is a proper coloring of $H$. Indeed, by definition of $\varphi'$ no obstruction can exist between nodes in $G$, and by definition of $\varphi$ no obstruction can exist between a vertex and their pending vertices, since the pending vertices are colored with colors $\{4, \dots, r\}$ if $v \in V(G) \setminus U$ and $[r] \setminus \varphi_U(v)$ if $v \in U$.

Conversely, suppose that $(H,n,\varphi)$ is a {\sc Yes}-instance of \rfix{r} and $\varphi''$ is a proper $r$-coloring of $H$ such that $|\varphi'' \ominus \varphi| \leq n$. Recall that in the coloring $\varphi$, each vertex of the copy of $G$ in $H$ has $n+1$ neighbors in color $i$ for each $i \in \{4, \dots, r\}$. Thus $\varphi''$ must use only colors $\{1,2,3\}$ on those vertices. Moreover, the vertices $u$ in the copy of $U$ in $H$ have $n+1$ neighbors in color $i$ for each $i \in [r] \setminus \{\varphi_U(u)\}$. Then the restriction of $\varphi''$ to $V(G)$ is a $3$-coloring of $G$ that satisfies $\varphi''(u) = \varphi_U(u)$ for each node $u\in U$.
\end{proof}

Since the described construction preserves both the planarity and the bipartiteness, we obtain the following.
\begin{thm} \label{thm:npc}
The \rfix{r}~problem is {\em NP}-complete for any $r\geq 3$, even if the input graph is planar and bipartite.
\end{thm}

This shows that fixing a given coloring remains hard, even if the number of available colors is much bigger that the chromatic number of the input graph.

The construction of the graph $G$ was based on the following observation.
\begin{obser}\label{obs-list}
Let $(G,k,\varphi)$ be a {\sc Yes}-instance of \fix and let $\varphi'$ be a witness.
If a vertex $v$ of $G$ has $k+1$ private neighbors in color $c$, then $\varphi'(v) \neq c$.
\end{obser}
This is equivalent to introducing lists of allowed colors. Thus for simplicity we will sometimes consider a generalization of \fix, where every vertex $v$ of the input graph has a list $L(v)$ of possible colors. We ask if $k$ recoloring operations are enough to transform a given initial coloring $\varphi$ into a proper \emph{list} coloring with lists $L$. We will denote this problem by \listfix.

\subsection{Exact algorithm for the \fix~problem} \label{sec:exact}

In this section we deal with the optimization version of the \fix~problem.
Note that the brute force algorithm works in time $(\sum_{k=0}^n \binom{n}{k} (r-1)^k) \cdot n^{\O(1)} = r^n \cdot n^{\O(1)}$. We shall obtain a better algorithm by reducing the instance of our problem to an instance of the so-called {\sc Max Weighted Partition} problem and then solve it, using the algorithm by Bj\"orklund, Husfeldt and Koivisto \cite{BHK}.
A {\em partition} of the set $N$ is a family of sets $S_1,\ldots,S_r$ such that $\bigcup_{i=1}^r S_i = N$ and $S_i \cap S_j = \emptyset$ for every $i \neq j$. Notice that we do not require for the sets $S_i$ to be non-empty.

\vskip 5pt
\noindent\begin{tabularx}{\textwidth}{| X |}
\hline
{\bf Problem: {\sc Max Weighted Partition}}\\
{\bf Instance:} A set $N$, integer $d$ and functions $f_1,f_2,\ldots,f_d \colon 2^N \to [-M,M]$  \\
for some integer $M$.\\
{\bf Question:} What is the maximum $w$, for which there exists a partition $S_1,S_2,\ldots,S_d$ such that $\sum_{i=1}^d f_i(S_i) = w$?\\
\hline
\end{tabularx}
\vskip 5pt

Let $G$ be a graph and let $\varphi$ be its $r$-coloring.
We shall construct a corresponding instance $\J = (N,d,f_1,\ldots,f_d)$ of {\sc Max Weighted Partition} problem.
Set $N = V(G)$ and $d=r$. We define functions $f_1,f_2,\ldots,f_d$ as:

$f_i(S) = \begin{cases} - |S \setminus \varphi^{-1}(i)| & \textrm{if $S$ is independent,} \\
						- r \cdot n & \textrm{otherwise.} \end{cases}$

In this way every partition of $V(G)$ into $r$ independent set, corresponding to the proper $r$-coloring $\varphi'$, has the total weight $\left(- \sum_{i=1}^r |\varphi'^{-1}(i) \setminus \varphi^{-1}(i)|\right)$. It is also easy to notice that any partition into independent sets has greater weight than any partition having at least one non-independent set.

The only thing left is to prove that the weight is maximized for a partition corresponding to a coloring $\varphi'$ such that $\dist(\varphi',\varphi)$ is minimum. To see this, notice that $\dist(\varphi',\varphi) = |\{v \in V \colon \varphi'(v) \neq \varphi(v)\}| = |\bigcup_{i=1}^r \{v \in V \colon \varphi'(v)=i \land \varphi(v) \neq i\}| = \sum_{i=1}^r|\{v \in V \colon \varphi'(v)=i \land \varphi(v) \neq i\}| = \sum_{i=1}^r|\varphi'^{-1}(i) \setminus \varphi^{-1}(i)|$.
Moreover, if a weight of a found partition is at most $-r \cdot n$, it contains at least one non-independent set, which means that $r < \chi(G)$ and therefore $\fixnumb{\varphi}{r}{G} = \infty$.

Now we can use the algorithm by Bj\"orklund {\em et al.} to find the optimal solution for $\J$.

\begin{thm}[Bj\"orklund, Husfeldt, Koivisto \cite{BHK}]
{\sc Max Weighted Partition} problem can be solved in time:
\begin{itemize}
\item $2^n d^2 M \cdot n^{\O(1)}$, using exponential space,
\item $3^n d^2 M\cdot n^{\O(1)}$, using polynomial space,
\end{itemize} where $n$ is the cardinality of the ground set.
\end{thm}

We can assume that $r \leq n$, since otherwise we can shift all colors down.
Since $d=r$ and  $M = n \cdot r$, we obtain the following corollary.

\begin{cor} \label{cor:exact-algo}
The optimization version of the \fix problem can be solved in time:
\begin{itemize}
\item $2^n \cdot n^{\O(1)}$, using exponential space,
\item $3^n \cdot n^{\O(1)}$, using polynomial space,
\end{itemize} where $n$ is the number of vertices in the input graph.
\end{cor}

It is known, that assuming the ETH, no $2^{o(n)} \cdot n^{\mathcal{O}(1)}$-algorithm for the \rcoloring{r} problem exists (see e.g. \cite{book}). Thus, by Proposition \ref{lem:equivalent} we immediately obtain the following corollary.

\begin{cor} \label{cor:exact-eth}
For any constant $r \geq 3$, there is no algorithm for the \rfix{r} problem with running time $2^{o(n)}$ (where $n$ is the number of vertices in the input graph), unless the ETH fails.
\end{cor}

This shows that the algorithms given by Corollary \ref{cor:exact-algo} are asymptotically optimal.

\section{Parameterized complexity of \fix problem} \label{sec:parameterized}

Since the \rfix{r} problem is computationally hard, we will turn to parameterized complexity theory, in hope to identify tractable cases.

\subsection{Parameterized by $k+r$}
Clearly for a fixed integer $k$, the problem can be easily solved in time $ \binom{n}{k} r^k \cdot n^{\O(1)} = n^k (r-1)^k \cdot n^{\O(1)}$. To do it, we have to consider every $k$-element subset of vertices and check whether recoloring the chosen vertices (in $(r-1)^k$ ways, since we are interested in recoloring {\em at most} $k$ vertices and thus some colors may remain unchanged) allows us to obtain a proper coloring. Therefore our problem is in XP, when parameterized by $k$. In the remainder of this section we show that the problem is in FPT, when parameterized by $k+r$, i.e., we can solve it in time $f(k,r) \cdot n^{\mathcal{O}(1)}$ (note that the degree of the polynomial function of $n$ does not depend on $k+r$).
Consider the following algorithm.

\begin{algorithm}[H]
\caption {Fix($r$, $\I=(G,k,\varphi)$)}
\lIf {$\varphi$ is a proper coloring of $G$} {\Return {\sc Yes}} \label{lineproper}
\lIf {$k=0$} {\Return {\sc No}} \label{k0}
$xy \gets $ any edge from $G^\varphi$ \label{choosexy}\\
\ForEach {$col \in [r] \setminus \{\varphi(x)\}$ \label{loopstart}}
{
        $\varphi_1 \gets \varphi$ with vertex $x$ recolored to $col$\\
        \lIf {$Fix(r, (G,k-1,\varphi_1))=$ {\sc Yes}}{\Return {\sc Yes}} \label{linecall}
}
\ForEach {$col \in [r] \setminus \{\varphi(y)\}$}
{
        $\varphi_1 \gets \varphi$ with vertex $y$ recolored to $col$\\
        \lIf {$Fix(r, (G,k-1,\varphi_1))=$ {\sc Yes}}{\Return {\sc Yes}}\label{linecall2}
}
\Return {\sc No} \label{lineNo}
\end{algorithm}

\begin{lem} \label{lem:step}
Let $\varphi$ be a non-proper $r$-coloring of $G$.
Then $\I=(G,k,\varphi)$ is a {\sc Yes}-instance of \rfix{r} if and only if for any edge $xy \in E(G^{\varphi})$ there exists an $r$-coloring (possibly non-proper) $\varphi_1$ of $G$ such that:
\begin{enumerate}
\item $\varphi \ominus \varphi_1 = \{x\}$ or $\varphi \ominus \varphi_1 = \{y\}$,
\item $\I'=(G,k-1,\varphi_1)$ is a {\sc Yes}-instance of \rfix{r}.
\end{enumerate}
\end{lem}

\begin{proof}
First assume that $\I=(G,k,\varphi)$ is a {\sc Yes}-instance of \rfix{r} and let $\varphi'$ be its witness. Consider an edge $xy$ from $G^{\varphi}$. By the definition of $G^{\varphi}$, we have $\varphi(x) = \varphi(y)$.
Since $\varphi'$ is proper, clearly $\varphi'(x) \neq \varphi'(y)$. Then at least one of the vertices $x,y$ has changed its color. Without loss of generality assume that $\varphi'(x) \neq \varphi(x)$. Let $\varphi_1$ be a coloring defined as follows.
$$\varphi_1(u) = \begin{cases}\varphi(u) & \text{ if } u \neq x\\
						\varphi'(u) & \text{ if } u = x.\end{cases}$$
It is clear that it satisfies the conditions given in lemma.

Now consider $\varphi$ being an $r$-coloring of $G$ and let $xy$ be an edge from $G^{\varphi}$. 
Without loss of generality let $\varphi_1$ to be some $r$-coloring of $G$ such that $\varphi \ominus \varphi_1 = \{x\}$ and the instance $\I'=(G,k-1,\varphi_1)$ is a {\sc Yes}-instance of \rfix{r}. 

Let $\varphi_1'$ be a witness of $\I'$. Notice that $\dist(\varphi_1',\varphi) \leq \dist(\varphi_1',\varphi_1) + \dist(\varphi_1,\varphi) \leq k$ and therefore $\I=(G,k,\varphi)$ is a {\sc Yes}-instance of \rfix{r} with witness $\varphi_1'$.
\end{proof}

\begin{lem}
The algorithm {\em Fix} solves \rfix{r} problem for any $r$.
\end{lem}

\begin{proof}
Let $\I = (G,k,\varphi)$ be an instance of \rfix{r}. If $\varphi$ is a proper labeling of $G$, then the algorithm returns {\sc Yes} in line \ref{lineproper}. Suppose then that $\varphi$ is not proper. If $k=0$, the algorithm returns {\sc No} in line \ref{k0}.

Assume that $k > 0$ and the algorithm works properly for all instances with parameter smaller than $k$. Suppose that $\I = (G,k,\varphi)$ is a {\sc Yes}-instance of \rfix{r}. Let $xy$ be an edge chosen in line \ref{choosexy}.
Then, by Lemma \ref{lem:step}, there exist an $r$-coloring $\varphi_1$ of $G$ such that $\varphi \ominus \varphi_1 = \{x\}$ or $\varphi \ominus \varphi_1 = \{y\}$ and $\I'=(G,k-1,\varphi_1)$ is a {\sc Yes}-instance of \rfix{r}. Without loss of generality assume that $\varphi \ominus \varphi_1 = \{x\}$. Let us consider an iteration of the loop in lines \ref{loopstart}--\ref{linecall} color $\varphi_1(x)$. By the inductive assumption, the recursive call in line \ref{linecall} returns {\sc Yes} and therefore the whole algorithm returns {\sc Yes}.

If $\I$ is a {\sc No}-instance of \rfix{r}, then by the inductive assumption and Lemma \ref{lem:step}, for every $col$ the recursive calls in lines \ref{linecall} and \ref{linecall2} return {\sc No}. Therefore the algorithm returns {\sc No} in line \ref{lineNo}.
\end{proof}

Let $T(n,k)$ be the computational complexity of the algorithm {\em Fix}. We can write the following recursive formula:
$$T(n,k) \leq n^{\O(1)} + (r-1) \cdot T(n,k-1) + (r-1) \cdot T(n,k-1).$$

Solving it, we obtain the following.

\begin{thm} \label{fpt-algo}
The algorithm {\em Fix} solves \fix problem in time $T(n,k) \leq  \left( 2(r-1) \right)^k \cdot n^{\mathcal{O}(1)}$.
\end{thm}

\begin{cor} \label{fpt}
The \fix problem is in FPT, when parameterized by $k+r$.
\end{cor}

Since $r \leq n$, we also obtain the following.

\begin{cor} \label{xp}
The \fix problem is in XP, when parameterized by $k$.
\end{cor}

%

A natural question to ask is whether the \fix problem is FPT, when parameterized by $k$ only (i.e., the number $r$ of colors is not fixed). We answer this question in the negative.

Let us start with presenting an auxiliary problem, called \msi. An instance of \msi is a host graph $H$, whose vertices are partitioned into $k$ color classes $V_1 \cupdot V_2 \cup \ldots \cupdot V_k$, and a pattern graph $P$ with $k$ vertices $\{u_1,u_2,\ldots,u_k\}$.
We ask if there exists a mapping $\psi \colon V(P) \to V(H)$, such that $\psi(u_i) \in V_i$ for all $i$, and $\psi(u_i)\psi(u_j) \in E(H)$ for all $u_iu_j \in E(P)$.
Observe that we can thus assume that each $V_i$ is an independent set in $G$, and that there is an edge between $V_i$ and $V_j$ if and only if $u_iu_j \in E(P)$.

The $W[1]$-hardness of \msi follows from the $W[1]$-hardness of {\sc Multicolored Clique}~(see \cite{book}). Marx \cite{Marx}, showed a stronger lower bound for this problem (see also Marx and Pilipczuk~\cite{MP}).

\begin{thm}[Marx \cite{Marx}]\label{thm:msi}
For any computable function $f$, \msi cannot be solved in time $f(k) \cdot (|V(H)| + |E(H)|)^{o(k / \log k)}$, unless the ETH fails, even if $P$ is 3-regular and bipartite.
\end{thm}

The reduction again exploits the close relation between fixing a coloring and extending a partial coloring.

\begin{thm}\label{thm:w1hard}
The \fix problem is $W[1]$-hard, when parameterized by the number $k$ of allowed recoloring operations.
Moreover, it cannot be solved in time $f(k) \cdot n^{o(k / \log k)}$ for any computable function $f$, unless the ETH fails, even if the input graph is bipartite.
\end{thm}

\begin{proof}
Consider an instance of \msi, consisting of the host graph $H$ and the 3-regular bipartite pattern graph $P$.
The vertex set of $P$ is $\{u_1,u_2,\ldots,u_k\}$, note that $P$ has $3k/2$ edges.
The vertex set of $H$ has $n$ elements and is partitioned into $k$ independent sets $V_1,V_2,\ldots,V_k$. 

To make the description more clear, we will first create an instance $(G,\ell,\varphi,L)$ of \listfix, where $\ell$ is the number of allowed recoloring operations, $\varphi$ is the initial coloring, and $L$ are lists of allowed colors in the final coloring.

The colors used by $\varphi$ and appearing on lists $L$ will be identified with the vertices of $H$, except for one additional special color, denoted by 0. Thus the total number of colors is $n+1$. Also set $\ell=4k$.

First, we introduce to $G$ vertices $x_1,x_2,\ldots,x_k$, each colored with color 0. For each $i$, we set $L(x_i) = V_i$.
Now, for every edge of $P$ we introduce the following edge gadget.
Consider an ordered pair $(i,j)$, such that $1 \leq i,j \leq k$ and $u_iu_j \in E(P)$. For each $v \in V_i$, we introduce a vertex $x_{i,j}^v$, which is adjacent to $x_i$. For all $v \in V_i,v' \in V_j$ the vertices $x_{i,j}^v$ and $x_{j,i}^{v'}$ are adjacent, so the main part of every edge gadget is a complete bipartite graph.
Finally, for every $x_{i,j}^v$ we set $\varphi(x_{i,j}^v) = v$ and $L(x_{i,j}^v) = \{v\} \cup \{ v' \in V_j \colon vv' \in E(H)\}$ (see Fig.\ref{fig:edgegadget}).
\begin{figure}[h]
\centering

\begin{tikzpicture}[scale = 1.3]
\node[draw,circle] (xi1) at (1,1){$v_1$};
\node[draw,circle] (xi2) at (1,0){$v_2$};
\node[draw,circle] (xi3) at (1,-1){$v_3$};

\node[draw,circle] (xj1) at (3,1){$v'_1$};
\node[draw,circle] (xj2) at (3,0){$v'_2$};
\node[draw,circle] (xj3) at (3,-1){$v'_3$};

\draw (1,0) ellipse (0.8 and 2);
\draw (3,0) ellipse (0.8 and 2);
\node at (1,1.5) {$V_i$};
\node at (3,1.5) {$V_j$};

\draw (xi1) -- (xj3);
\draw (xi2) -- (xj1);
\draw (xi2) -- (xj1);
\draw (xi2) -- (xj2);
\draw (xi3) -- (xj3);
\end{tikzpicture} \hfill
\begin{tikzpicture}[scale = 1.2, yscale=2]

\node[draw,circle] (xi) at (0,0){$0$};

\node[draw,circle] (xi1) at (1,1){$v_1$};
\node[draw,circle] (xi2) at (1,0){$v_2$};
\node[draw,circle] (xi3) at (1,-1){$v_3$};

\node[draw,circle] (xj1) at (3,1){$v'_1$};
\node[draw,circle] (xj2) at (3,0){$v'_2$};
\node[draw,circle] (xj3) at (3,-1){$v'_3$};

\node[draw,circle] (xj) at (4,0){$0$};

\node [above left  =0.01 of xi] {$v_1,v_2,v_3$};
\node [left  =0.01 of xi1] {$v_1,v_3'$};
\node [below  =0.01 of xi2] {$v_2,v_1',v_2'$};
\node [left  =0.01 of xi3] {$v_3,v_3'$};

\node [below right  =0.01 of xj] {$v'_1,v'_2,v'_3$};
\node [right  =0.01 of xj1] {$v_1',v_2$};
\node [above  =0.01 of xj2] {$v_2',v_2$};
\node [right  =0.01 of xj3] {$v_3',v_1,v_3$};

\draw (xi) -- (xi1);
\draw (xi) -- (xi2);
\draw (xi) -- (xi3);
\draw (xj) -- (xj1);
\draw (xj) -- (xj2);
\draw (xj) -- (xj3);

\draw (xi1) -- (xj1);
\draw (xi1) -- (xj2);
\draw (xi1) -- (xj3);
\draw (xi2) -- (xj1);
\draw (xi2) -- (xj2);
\draw (xi2) -- (xj3);
\draw (xi3) -- (xj1);
\draw (xi3) -- (xj2);
\draw (xi3) -- (xj3);
\end{tikzpicture}
\caption{The sets $V_i, V_j$ (left) and the corresponding edge gadget used in the proof of Theorem \ref{thm:w1hard} (right, the numbers in circles denote the coloring $\varphi$, and the number beside circles denote lists $L$).}
\label{fig:edgegadget}
\end{figure}
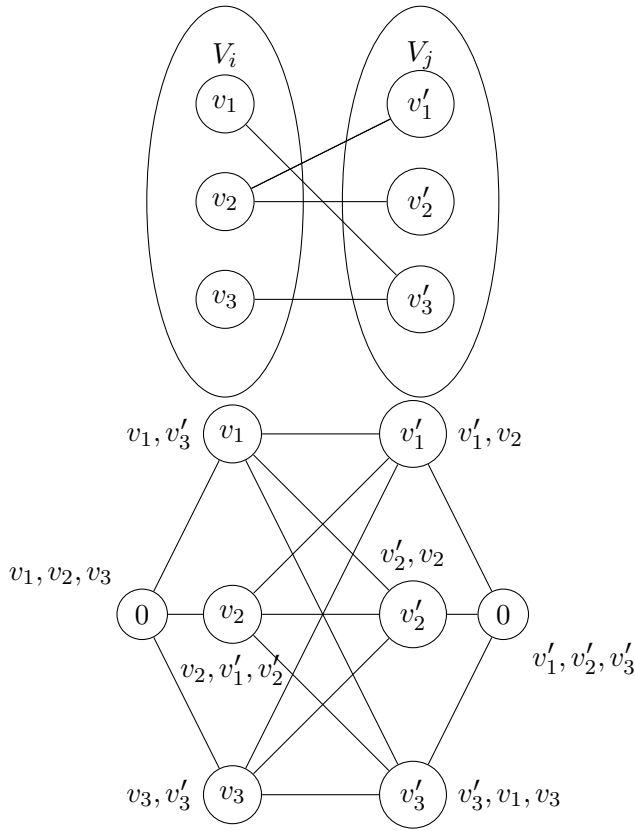

Note that the total number of vertices in $G$ is at most $k + \sum_{u_iu_j \in E(P)} (|V_i| + |V_j|) = k + 3\sum_{i \in V(P)} |V_i| = k + 3n = O(n)$. We claim  that $\varphi$ is at distance at most $\ell=4k$ for some proper list coloring of $G$ with lists $L$ if and only if $(H,P)$ is a {\sc Yes}-instance of \msi.

First, assume that there is a mapping $\psi \colon V(P) \to V(H)$, such that $\psi(u_i) \in V_i$ for any $u_i \in V(P)$, and $\psi(u_i)\psi(u_j) \in E(H)$ for any $u_iu_j \in E(P)$. 
Let $\varphi'$ be the coloring of $G$ defined as follows:
\[
\varphi'(x) = 
\begin{cases}
\psi(u_i) &\text{ if } x = x_i\\
\psi(u_j) &\text{ if } x = x_{i,j}^v \text{ and } \psi(u_i)=v\\
\varphi(x) &\text{ otherwise}.
\end{cases}
\]

First, note that $\varphi'$ respects lists $L$. Clearly $\psi(u_i) \in V_i = L(x_i)$ for every $1 \leq i \leq k$.
Consider $(i,j)$, such that $u_iu_j \in E(P)$ and let $\psi(u_i)=v$ and $\psi(u_j)=v'$. Then $\psi(u_j) = v'$ is a neighbor of $v$ in $H$, which implies that $v' \in L(x_{i,j}^v)$. Finally, we have $\varphi(x_{i,j}^v) = v \in L(x_{i,j}^v)$.

Now suppose that $\varphi'$ is not a proper coloring, i.e., there is an edge of $G$, whose endpoints get the same color.
First, assume that this edge is of the form $x_ix_{i,j}^v$ and $\varphi'(x_i) = \varphi'(x_{i,j}^v) = v'$. If $v = v'$, then we know that $x_{i,j}^v$  was initially colored with $v$, but was recolored to some color from $V_j$, so such a conflict is not possible. If $v \neq v'$, then $x_{i,j}^v$ did not change its color, so $\varphi'(x_{i,j}^v) = v \neq v'$, a contradiction.

So assume that we have $\varphi'(x_{i,j}^v) = \varphi'(x_{j,i}^{v'})$. Observe that this means, that exactly one of $x_{i,j}^v,x_{i,j}^{v'}$ was recolored, since otherwise they have colors from different sets $V_i,V_j$. Suppose $x_{i,j}^v$ was recolored to the color $v'$, which means that $\psi(u_i)=v$ and $\psi(u_j)=v'$. However, by the definition of $\varphi'$, the vertex $x_{j,i}^{v'}$ was recolored to obtain the color $v$, a contradiction.

Finally, observe that to obtain $\varphi'$, we recolored each vertex $x_i$ and two vertices from each edge gadget, which gives $k + 3/2 \cdot k \cdot 2 = 4k = \ell$ recoloring operations in total.

Now let us show that the converse implication holds. Suppose that there exists a proper coloring $\varphi'$ of $G$ with lists $L$, such that $|\varphi \ominus \varphi'|\leq \ell$. Observe that for every $x_i$ we have $\varphi(x_i) \neq \varphi'(x_i)$, otherwise $\varphi'$ does not respect lists $L$. Moreover, if $\varphi'(x_i) = v$ for some $v \in V_i$, then for every $j$ the vertex $x_{i,j}^v$ has to be recolored as well. This means that the minimum number of recoloring operations is $k + 3/2 \cdot k \cdot 2 = 4k = \ell$, and this is possible only if $\varphi \ominus \varphi'$ consists of vertices $x_1,x_2,\ldots,x_k$ and exactly two vertices from every edge gadget.

We claim that $\psi \colon V(P) \to V(H)$, defined by $\psi(u_i) = \varphi'(x_i)$ is the desired subgraph isomorphism from $P$ to $H$. By the choice of lists $L$, we have $\psi(u_i) \in V_i$. So suppose that we have an edge $u_iu_j$ of $P$, such that $\psi(u_i) = v$ and $\psi(u_j)=v'$, and $vv' \notin E(H)$. Observe that since $\varphi'(x_i) = \psi(u_i)$ is $v$, we have $\varphi(x_{i,j}^v) \neq \varphi'(x_{i,j}^v)$ and analogously $\varphi(x_{j,i}^{v'}) \neq \varphi'(x_{j,i}^{v'})$. Moreover, no other vertex from the edge gadget corresponding to $u_iu_j$ was recolored to obtain $\varphi'$. This is only possible if $\varphi'(x_{i,j}^v)=v'$ and $\varphi'(x_{j,i}^{v'})=v$, which in turn implies that $vv' \in E(H)$, since $\varphi'$ respects lists $L$. 

To transform the graph $G$ with lists $L$ and initial coloring $\varphi$ into an instance of the \fix problem we will use Observation \ref{obs-list}. For every vertex $x$ of $G$ and every color $c$ \emph{not appearing} in $L(x)$, we introduce $k+1$ vertices, colored $c$, which are private neighbors of $x$. Let $G'$ denote the obtained graph.

Observe that since $P$ is bipartite, $G'$ is bipartite as well.
The total number of vertices of $G'$ is $O(n \cdot k \cdot n) = O(n^3)$. Summing up, an algorithm solving the \fix problem in time $f(\ell) \cdot |V(G')|^{o(\ell / \log \ell)}$ could be used to decide if $(H,P)$ is a {\sc Yes}-instance of \msi in time  = $f(4k) \cdot (n^4)^{o(3k / \log 4k)} = f'(k) \cdot n^{o(k / \log k)}$, which would in turn contradict the ETH.
\end{proof}

Since $r \leq n$, we obtain the following corollary lower bound, which almost matches the bound from Theorem \ref{fpt-algo}.

\begin{thm}\label{cor:almost-optimal-k}
The \fix problem with $r$ colors and $k$ recoloring operations cannot be solved in time $f(k) \cdot r^{o(k / \log k)}$ for any computable function $f$, unless the ETH fails, even if the input graph is bipartite.
\end{thm}

\subsubsection{Kernelization}

Corollary \ref{fpt} yields that the \fix problem (parameterized by $k+r$) admits a kernel of size $\left( 2(r-1) \right)^k$ (see \cite[Chapter 2]{book} for more information about kernels). The next theorem shows, that a polynomial kernel for this problem cannot be obtained (under some standard complexity assumptions).

\begin{thm} \label{no-poly-kernel}
For any $r\geq 3$, the \rfix{r} parameterized by the number $k$ of allowed recoloring operations, does not admit a polynomial kernel, unless {\em NP} $\subseteq$ {\em coNP / poly}.
\end{thm}

\begin{proof}

Let $t,n>0$ be two positive integers and let $\mathcal{S} = \{ (G_i, U_i, \varphi_{U_i})\}_{i \in [2^t]}$ be a set of $2^t$ instances of \preext{3} such that  $G_i$ is bipartite and $|V(G_i)| = n$, for all $i \in [2^t]$. Let $k = t + n + 8$. Our reduction constructs from $\mathcal{S}$ a graph $H$ and a (non-proper) coloring $\varphi$ of $H$. This construction satisfies that $(H', k, \varphi)$ is a {\sc Yes}-instance of \rfix{r} if and only at least one of the instances in $\mathcal{S}$ is a {\sc Yes}-instance of \preext{3}. 
Since \preext{3} on bipartite graphs is NP-complete \cite{Kratochvil1993}, by Theorem \ref{thm-no-kernel-framework}, this reduction implies that \rfix{r} parameterized by the number of allowed recoloring operations does not admit a polynomial kernel, unless NP $\subseteq$ coNP / poly.

We will again start with creating an instance $(H,k,\varphi,L)$ of \listfix and then we will use Observation \ref{obs-list} to transform it into an equivalent instance of \fix.

The construction of $H$ and $\varphi$ considers two gadgets: the {\it problem selector gadget} and the {\it instance gadget}. 
The {\it problem-selector} gadget consists of a complete rooted binary tree of depth $t$, where the root of the tree has list $\{2,3\}$ and the remaining vertices have lists $\{1,2,3\}$. The root of the tree is colored $1$ in $\varphi$, and the rest of the nodes are $3$-colored in such a way that  the coloring restricted to the nodes in the tree is proper, and the children of a node have different colors in $\varphi$ (see Figure \ref{fig:problemselectorgadget} for an example when $t=3$).
Any proper recoloring obtained from $\varphi$  after at most $t+1$ recoloring operations corresponds to a path from the root to a leaf, where each node in the path takes the color of its child also belonging to the path, except for the leaf that is free to pick any color different than its initial one. The leaf picked to be the one that changes its color is called the {\it selected leaf}. The idea is to attach to each leaf an {\it instance gadget}, one for each $(G,U, \varphi_U) \in \mathcal{S}$ (note that the number of leaves is $2^t$, which is equal to $|\mathcal{S}|$). 
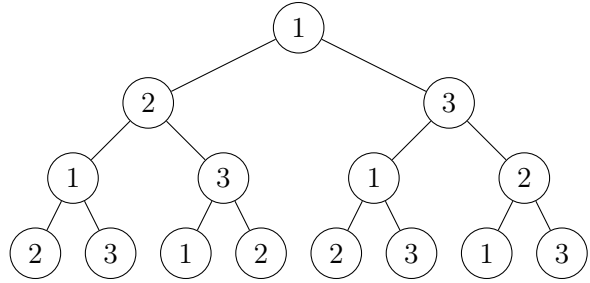
\begin{figure}
\centering
\begin{tikzpicture}[scale = 1]
\node [draw,circle] (v1) at (-1.5,8.5) {1};
\node [draw,circle] (v2) at (-3.5,7.5) {2};
\node [draw,circle] (v5) at (0.5,7.5) {3};
\node [draw,circle] (v3) at (-4.5,6.5) {1};
\node [draw,circle] (v4) at (-2.5,6.5) {3};
\node [draw,circle] (v6) at (-0.5,6.5) {1};
\node [draw,circle] (v7) at (1.5,6.5) {2};
\draw  (v1) edge (v2);
\draw  (v2) edge (v3);
\draw  (v2) edge (v4);
\draw  (v1) edge (v5);
\draw  (v5) edge (v6);
\draw  (v5) edge (v7);
\node [draw,circle] (v8) at (-5,5.5) {2};
\node [draw,circle] (v9) at (-4,5.5) {3};
\node [draw,circle] (v10) at (-3,5.5) {1};
\node [draw,circle] (v11) at (-2,5.5) {2};
\node [draw,circle] (v12) at (-1,5.5) {2};
\node [draw,circle] (v13) at (0,5.5) {3};
\node [draw,circle] (v14) at (1,5.5) {1};
\node [draw,circle] (v15) at (2,5.5) {3};
\draw  (v8) edge (v3);
\draw  (v9) edge (v3);
\draw  (v10) edge (v4);
\draw  (v11) edge (v4);
\draw  (v12) edge (v6);
\draw  (v13) edge (v6);
\draw  (v14) edge (v7);
\draw  (v15) edge (v7);
\end{tikzpicture}
\caption{A {\it problem Selector Gadget}, when $t = 3$. The number in the nodes represent the color of the corresponding node for the coloring $\varphi$. The lists of all vertices except for the root are $\{1,2,3\}$, while the list of the root is $\{2,3\}$.}
\label{fig:problemselectorgadget}
\end{figure}

The construction of the {instance gadget} is slightly more complicated. Let $(G,U, \varphi_U)$ be an instance of \preext{3} contained in $\mathcal{S}$, i.e., $G = (X \cup Y, E)$ is a bipartite graph and $U$ is a subset of $V(G)$, colored according to $\varphi_U$. For $j\in [3]$, we define $U_j := \{ u \in U : \varphi_U(u) = j\}$. The instance gadget of $(G,U, \varphi_U)$ is constructed as follows (see  Figure \ref{fig:instancegadget}): pick a copy of $G$ and add ten more nodes called $\{v_1, \dots, v_{10}\}$ with edges $v_1v_2$, $v_1v_3$, $v_2v_5$, $v_2v_6$, $v_3v_4$, $v_4v_7$, and $v_4v_8$. Moreover, we add all edges $v_5u,v_7u$ for $u \in U_2 \cup U_3$, all edges $v_6u,v_8u$ for $u \in U_1 \cup U_3$, and finally all edges $v_9u,v_{10}u$ for $u \in U_1 \cup U_2$.
The lists are set as follows: $L(v_1)=L(v_3)=\{1,2\}$; $L(v_2)=L(v_6)=L(v_8)=\{2,3\}$; $L(v_5)=L(v_4)=L(v_7)=\{1,3\}$; $L(v_9)=L(v_{10})=\{3\}$.
Finally, color all the nodes in $X$, $v_1$ and $v_4$ with color $1$, all the nodes in $Y$, $v_2$ and $v_3$ with color $2$, and $v_5$, $v_6$, $v_7$, $v_8$, $v_9$, $v_{10}$ with color $3$. 
Observe that the only colors apperaing in $\varphi$ and in the lists $L$ are $\{1,2,3\}$, so we will never use any other color.

Notice that the instance gadget is properly colored. Suppose that for some reason (this will be forced by the problem-selector gadget) the vertex $v_1$ must be recolored to a color different than $1$.  We will show that in that case, we can obtain a proper coloring of the nodes in the instance gadget changing colors of at most $n + 8$ vertices if and only if $(G,U, \varphi_U)$ is a {\sc Yes}-instance of \preext{3}, where $n = |G|$. Indeed, since the color 3 is forbidden to $v_1$ and it must change its color, it must get color $2$. This implies that also $v_2$ and $v_3$ must be recolored. Since the color 1 is forbidden for $v_2$ and the color 3 is forbidden for $v_3$, $v_2$ must change to the color $3$, and $v_3$ to the color $1$. Following the same arguments, $v_4$ must be recolored $3$, and hence $v_5$ and $v_7$ must be recolored to $1$, and $v_6$ and $v_8$ must be recolored to $2$. Note that $v_9$ and $v_{10}$ must stay in color $3$ since this is the only color in their lists. Notice that this situation implies that all nodes in $U_i$ must be recolored $i$, for $i\in \{1,2,3\}$. Therefore, the instance gadget can be properly recolored, if $G$ can be properly colored respecting the pre-coloring of $U$, i.e. if $(G, U, \varphi_U)$ is a {\sc Yes}-instance of \preext{3}. 
Note that the number of recoloring operations performed is at most $n+8$.

\begin{figure}[h]
\centering
\begin{tikzpicture}[scale = 0.8]

\draw  (-0.5,4) ellipse (1.5 and 3.5);

\draw[thick, dotted]  (-5,6.5) rectangle (0,5.5);
\draw[thick, dotted]  (-5,4.5) node (v9) {} rectangle (0,3.5);
\draw[thick, dotted]  (-5,2.5) rectangle (0,1.5);

\draw  (-4.5,4) ellipse (1.5 and 3.5);

\node (v11) at (-7,6.7){$v_5$};
\node[draw,circle] (v1) at (-7,6) {$3$};
\node [left  =0.01 of v1] {$1,3$};

\node (v11) at (-9.2,4.7){$v_6$};
\node[draw,circle] (v2) at (-9,4) {$3$};
\node [below  =0.01 of v2] {$2,3$};

\node (v11) at (-7,1.7){$v_9$};
\node[draw,circle]  (v3) at (-7,1){$3$};
\node [left  =0.01 of v3] {$3$};

\node (v11) at (2,6.7){$v_7$};
\node[draw,circle] (v4) at (2,6) {$3$};
\node [right  =0.01 of v4] {$1,3$};

\node (v11) at (4.2,4.7){$v_8$};
\node[draw,circle] (v5) at (4,4) {$3$};
\node [below  =0.01 of v5] {$2,3$};

\node (v11) at (2,1.7){$v_{10}$};
\node[draw,circle] (v6) at (2,1) {$3$};
\node [right  =0.01 of v6] {$3$};

\node (v11) at (-9.3,7.5) {$v_2$};
\node[draw,circle] (v7) at (-8.5,7.5) {$2$};
\node [above  =0.01 of v7] {$2,3$};

\node (v11) at (4.3,7.5) {$v_4$};
\node[draw,circle] (v8) at (3.5,7.5) {$1$};
\node[above  =0.01 of v8] {$1,3$};

\node (v11) at (1.5,8.3) {$v_3$};
\node[draw,circle] (v9) at (1.5,9) {$2$};
\node[above  =0.01 of v9] {$1,2$};

\node (v11) at (-2.5,11.2) {$v_1$};
\node[draw,circle] (v10) at (-2.5,10.5) {$1$};
\node[below  =0.01 of v10] {$1,2$};

\node at (-2.5,6) {\large $U_1$};
\node at (-2.5,4) {\large $U_2$};
\node at (-2.5,2) {\large $U_3$};

\draw  (v10) edge (v7);
\draw  (v7) edge (v1);
\draw  (v7) edge (v2);
\draw  (v4) edge (v8);
\draw  (v5) edge (v8);
\draw  (v8) edge (v9);
\draw  (v9) edge (v10);
\node (v11) at (-5,4) {};
\node (v13) at (-5,6) {};
\node (v12) at (-5,2) {};
\node (v16) at (0,2) {};
\node (v14) at (0,4) {};
\node (v15) at (0,6) {};

\draw[thick, dotted]  (v1) edge (v11);
\draw[thick, dotted]  (v1) edge (v12);
\draw[thick, dotted]  (v2) edge (v13);
\draw[thick, dotted]  (v2) edge (v12);
\draw[thick, dotted]  (v3) edge (v11);
\draw[thick, dotted]  (v3) edge (v13);
\draw[thick, dotted]  (v6) edge (v14);
\draw[thick, dotted]  (v6) edge (v15);
\draw[thick, dotted]  (v5) edge (v16);
\draw[thick, dotted]  (v4) edge (v14);
\draw[thick, dotted]  (v4) edge (v16);

\draw [thick, dotted] (v5) edge (v15);

\node at (-4.5,0) {\Large $X$};
\node at (-0.5,0) {\Large $Y$};
\end{tikzpicture}
\caption{Representation of an {\it instance gadget} of $G = (X \cup Y, E)$, with precolored set $U = U_1 \cup U_2 \cup U_3$. The numbers in the nodes represent their color in the initial coloring. The numbers next to the nodes denote the lists of allowed colors.}
\label{fig:instancegadget}
\end{figure}
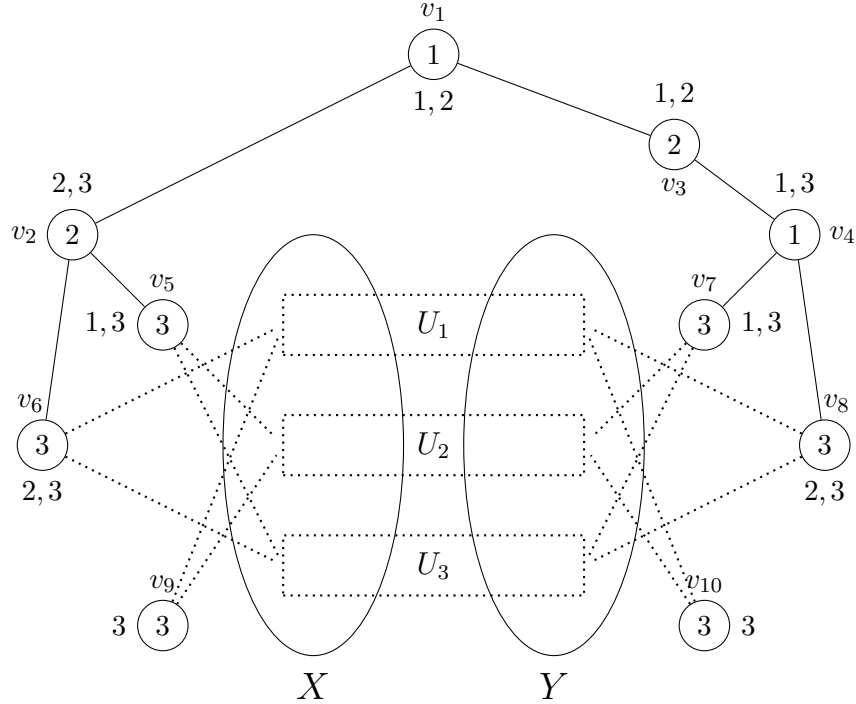

Then, we construct $H$ and $\varphi$ as follows: we start with building the problem-selector gadget. Then we identify each leaf colored $1$ with the node $v_1$ in a copy of the instance gadget built from one of the instances in $\mathcal{S}$ (we have an instance gadget for each instance). For the leaves colored $2$ (resp. $3$), we make an analogous version of the instance gadget, where the nodes colored $1$ change to $2$ (resp. $3$), the nodes colored $2$ change to $3$ (resp. $1$), the nodes colored $3$ change to $1$ (resp. $2$), and where $U_1, U_2, U_3$ are changed to $U_2, U_3, U_1$ (resp. $U_3, U_1, U_2$). Then we identify each leaf colored $2$ or $3$ with the node $v_1$ in a copy of the corresponding version of the instance gadget, built from one of the instances in $\mathcal{S}$.  Note that the size of $H$ is $2^t\cdot (n + 12(k+1) + 11) + k = \mathcal{O}(|S|(n + \log|S|))$, and it can be constructed in polynomial time in $|S|$ and $n$. 

We claim that $(H, t+n+8, \varphi)$ is a {\sc Yes}-Instance of \rfix{r} if and only if there exists a {\sc Yes}-instance of \preext{3} in $\mathcal{S}$. Indeed, if a recoloring of $H$ from the initial coloring $\varphi$ forces (by the properties of the problem selector gadget) to pick a path  from the root of the binary tree to a leaf, hence recoloring at least $t$ vertices. Recoloring the selected leaf {\it activates} the corresponding instance gadget, where the recoloring with at most $n+8$ changes implies that the underlying instance of \preext{3} is a {\sc Yes}-instance. In the converse, if there exists a {\sc Yes}-instance of \preext{3} in $\mathcal{S}$, then the coloring of the corresponding instance gadget (requiring to recolor at most $n+8$ vertices), and later a path from the corresponding leaf to the root (requiring to recolor $t$ vertices), is a witness for the instance $(H, t+n+8, \varphi)$ of \rfix{r}. 
%
\end{proof}

Observe that if all instances in $\mathcal S$ are bipartite, then also $H$ is bipartite. Thus we obtain the following.

\begin{cor}
For any $r \geq 3$, the \rfix{r} problem parameterized by $k$, does not admit a polynomial kernel, unless {\em NP} $\subseteq$ {\em coNP / poly}, even if the input graph is bipartite.
\end{cor}

\subsection{Parameterized by $r$ + the treewidth of the input graph} \label{sec:treewidth}

In this section we consider the optimization version of \rfix{r} problem for graphs with bounded treewidth. For more information about tree decompositions and treewidth, the reader is referred to Diestel's book \cite{Diestel}. 

Using a standard dynamic programming (see the survey by Bodlaender and Koster \cite{BK} for some examples) on a tree decomposition, one can obtain the following result.

\begin{thm} \label{thm:tw}
For any fixed $r$, the optimization version of \rfix{r} problem can be solved in time $r^t \cdot n^{\O(1)}$, where $n$ is the number of vertices of the input graph and $t$ is its treewidth.
\end{thm}

We skip the proof as the approach is very standard. It is enough to store the set of all proper $r$-coloring of every bag of the nice tree decomposition of the input graph.

Lokshtanov {\em et al.} \cite{LMS} have proven that, assuming the SETH, the simple $r^t \cdot n^{\mathcal{O}(1)}$-time algorithm for the $r$-coloring of an $n$-vertex graph of treewidth at most $t$ is asymptotically optimal.

\begin{thm}[Lokshtanov {\em et al.} \cite{LMS}]
For any $r \geq 3$, the \rcoloring{r} problem cannot be solved in time $(r - \epsilon)^t \cdot n^{\mathcal{O}(1)}$, where $n$ is the number of vertices of the input graph and $t$ is its treewidth, unless the SETH fails.
\end{thm}

By Proposition \ref{lem:equivalent}, we observe that  the algorithm given by Theorem \ref{thm:tw} is asymptotically optimal as well.

\begin{cor} \label{cor:tw-seth}
For any fixed $r \geq 3$ and any $\epsilon > 0$, there is no algorithm for the  \rfix{r} problem with running time $(r-\epsilon)^{t} \cdot n^{\mathcal{O}(1)}$, where $n$ is the number of vertices of the input graph and $t$ is its treewidth, unless the SETH fails.
\end{cor}

\section{Fixing number} \label{sec:number}

Recall that for a graph $G$ and its $r$-coloring $\varphi$, by $\fixnumb{\varphi}{r}{G}$ we denote the minimum number vertices that have to be recolored to obtain some proper $r$-coloring of $G$.

An {\em $r$-fixing number} of a graph $G$ (denoted by $\fixnumb{}{r}{G}$) is a maximum value of $\fixnumb{\varphi}{r}{G}$ over all colorings $\varphi \colon V(G) \to \{1,..,r\}$.
\begin{definition}By $\fixnumber{G}$ we denote the {\em fixing number} of a graph $G$, defined as a maximum value of $\fixnumb{}{r}{G}$ over all $r \geq \chi(G)$.
\end{definition}

\begin{lem} \label{lem:identify}
Let $\varphi$ be some $r$-coloring of $G$ and $\varphi'$ be an $(r+1)$-coloring of $G$ such that $\varphi^{-1}(i) = \varphi'^{-1}(i)$ for $i \in \{1,\ldots,r-1\}$. Then $\fixnumb{\varphi'}{r+1}{G} \leq \fixnumb{\varphi}{r}{G}$.
\end{lem}

\begin{proof}
Recoloring the vertices in the same way as with $\varphi$ makes $\varphi'$ proper.
\end{proof}

\begin{lem}
For all graphs $G$ and $r \geq \chi(G)$ it holds that $\fixnumb{}{r+1}{G} \leq \fixnumb{}{r}{G}$.
\end{lem}

\begin{proof}
Let $\varphi'$ be an $(r+1)$-coloring of $G$ such that $\fixnumb{\varphi'}{r+1}{G}= \fixnumb{}{r+1}{G}$. Let $\varphi$ be an $r$-coloring of $G$ obtained from $\varphi'$ by identifying colors $r$ and $r+1$. By Lemma \ref{lem:identify} we obtain the following.
$\fixnumb{}{r+1}{G} = \fixnumb{\varphi'}{r+1}{G} \leq \fixnumb{\varphi}{r}{G} \leq \fixnumb{}{r}{G}$.
\end{proof}

\begin{cor} \label{cor:rischi}
For all graphs $G$ it holds that $\fixnumber{G} = \fixnumb{}{\chi(G)}{G}$.
\end{cor}

Let $G$ be a bipartite graph with bipartition classes $X,Y$ and let $\varphi$ be its 2-coloring.
Recall from Observation \ref{bipartite} that $$\fixnumb{\varphi}{r}{G} = \sum_{ \substack{ C \colon \text{connected} \\ \text{component of }G}} 
\min \{ |\left ( X \ominus \varphi^{-1}(1) \right ) \cap V(C) |,| \left ( X \ominus \varphi^{-1}(2) \right ) \cap V(C)| \}.$$

Note that if $|\left ( X \ominus \varphi^{-1}(1) \right ) \cap V(C) | \geq |C|/2$, then $| \left ( X \ominus \varphi^{-1}(2) \right ) \cap V(C)| \leq |C|/2$ for any connected component $C$ of $G$. Thus we can easily obtain the following corollary.

\begin{cor} \label{cor:bipartite-upper}
$\fixnumber{G} \leq \lfloor n/2 \rfloor$ for every bipartite graph $G$ on $n$ vertices.
\end{cor}

This result can be generalized for non-bipartite graphs.

\begin{thm} \label{thm-fixupper}
For all $G$ it holds that $\fixnumber{G} \leq \left  \lfloor n\cdot\frac{\chi(G)-1}{\chi(G)} \right \rfloor$.
\end{thm}

\begin{proof}
For a graph $G = (V,E)$ set $r = \chi(G)$ and let $\varphi$ an $r$-coloring of $G$ such that $\fixnumb{\varphi}{r}{G} = \fixnumber{G}$.
Consider some proper $r$-coloring $\varphi'$ of $G$. Let $A_i = \varphi^{-1}(i)$ and $A'_i = \varphi'^{-1}(i)$ for all $i \in [r]$.
By $B_{i,j}$ we denote $A_i \cap A'_j$. Clearly $\bigcup_{i,j} B_{i,j} = V$. Note that for any permutation $\sigma$ of $[r]$, we can obtain a proper $r$-coloring of $G$ from $\varphi$ by recoloring all vertices but $C_\sigma = \bigcup _{i \in [r]} B_{i, \sigma(i)}$ (this proper coloring will be equivalent to $\varphi'$ up to the permutation of colors).

Suppose that $|C_\sigma| < \frac{n}{r}$ for all $\sigma$. On one hand we have $| \bigcup_\sigma C_\sigma| \leq \sum_\sigma |C_\sigma| < r! \cdot \frac{n}{r}$. On the other hand, we have:

\begin{align*}
|\bigcup_\sigma C_\sigma| = & |\bigcup_\sigma \bigcup_i B_{i,\sigma(i)}| = |\bigcup_i \bigcup_\sigma B_{i,\sigma(i)}|=|\bigcup_i \bigcup_j \bigcup_{\substack{\sigma \text{ s.t. } \\ \sigma(i)=j}} B_{i,j}|&\\
=&  (r-1)! |\bigcup_i \bigcup_j B_{i,j}| = (r-1)! n.
\end{align*}

This is a contradiction, so there exists $\sigma$ with $|C_\sigma| \geq \frac{n}{r}$ and therefore we can obtain a proper $r$-coloring of $G$ by recoloring at most $n \cdot \frac{r-1}{r}$ vertices. Since $\fixnumber{G}$ is an integer, we obtain our claim.
\end{proof}

To see  that this bound is attainable, consider a graph $G(m,r)$ on $n = m \cdot r$ vertices, consisting of $m$ disjoint copies of $K_r$. Clearly $\chi(G(m,r)) = r$. Let $\varphi$ be an $r$-coloring of $G(m,r)$ such that $\varphi(v)=1$ for every vertex $v$.
Clearly we have to recolor every vertex but one from every copy of $K_r$, which gives us $\fixnumber{G(m,r)} \geq \fixnumb{\varphi}{r}{G(m,r)} = m(r-1) = n \; \frac{r-1}{r}$.

However, there are graphs $G$ for which the value of $\fixnumber{G}$ is significantly smaller. For example, consider an odd cycle $C_n$ for $n \geq 9$. Clearly $\chi(C_n) = 3$. Let $\varphi$ be any coloring of $C_n$ with $r \geq 3$ colors. Arbitrarily choose vertex $v$ and remove it from $C_n$, obtaining a path $P_{n-1}$. Since $\chi(P_{n-1})=2$, then by Theorem \ref{thm-fixupper}, we can obtain a proper coloring of $P_{n-1}$ by recoloring at most $\lfloor (n-1)/2 \rfloor$ vertices. Then we can restore vertex $v$ and, if necessary, recolor it to an available color (there is always at least one). In this way we performed at most $1 + \lfloor (n-1)/ 2 \rfloor$ recoloring operations, which is roughly $\frac{n}{2}$ compared to $\frac{2n}{3}$ given by Theorem \ref{thm-fixupper}.

Another direction of research is to find lower bounds for the fixing number. Let us start with considering a star.

\begin{prop} \label{prop:star-lower}
Let $H$ be a star with $k$ leaves. Then there exists a 2-coloring $\varphi$, such that:
$\fixnumb{\varphi}{2}{H} \geq \lceil k/2 \rceil$.
\end{prop}
\begin{proof}
Partition the set of leaves of $H$ into 2 subsets $L_1,L_2$, such $|L_1| = \lceil k/2 \rceil$ and $\lfloor k/2 \rceil$.
Color the vertices from each group $L_i$ with the color $i$, and the root of $H$ with the color 1.
Then, to make the coloring proper, we either have to recolor at least $\lceil k/2 \rceil$ leaves, or at least $\lfloor k/2 \rfloor$ leaves and the root.

If $k$ is even, then \[\min ( \lceil k/2 \rceil , \lfloor k/2 \rfloor +1 ) = \min ( k/2 , (k-1)/2 +1 ) = k/2 = \lceil k/2 \rceil.\]
On the other hand, if $k$ is odd, then
\[\min ( \lceil k/2 \rceil , \lfloor k/2 \rfloor +1 ) = \min ( (k+1)/2, (k-1)/2+1  ) = (k+1)/2 = \lceil k/2 \rceil.\]
\end{proof}

Proposition~\ref{prop:star-lower} can be generalized to all graphs in the following way.

\begin{prop} \label{prop:lower}
Let $G$ be a connected graph with $n \geq 2$ vertices. Then $\fixnumber{G} \geq \lfloor n/2 \rfloor$.
\end{prop}

\begin{proof}
Let $H$ be a spanning tree of $G$  We will construct a 2-coloring $\varphi$ of $H$, such that $\fixnumb{\varphi}{2}{H} \geq \lfloor n/2 
\rfloor$. 
First, observe that if $H$ is a star, then such a coloring can be obtained using Proposition~\ref{prop:star-lower}.

We do induction on the number $n$ of vertices in $H$. If $n \leq 3$, then $H$ is a star, so the claim holds.
So assume that $n \geq 4$ and $H$ is not a star.
Let $H$ be rooted in an arbitrary vertex $w$ and let $u_1$ be a leaf of $H$ with the maximum distance from $w$. Let $v$ be the parent of $u_1$ and $U=\{u_1,u_2,\ldots,u_\ell\}$ be the children of $v$. Clearly all vertices in $U$ are leaves in $H$. Moreover, since $H$ is not a star, there are at least two vertices (including $w$) that are not in $\{v\} \cup U$.

Consider two cases. If $\ell$ is even, then let $H'$ be the graph obtained from $H$ by removing all vertices from $U$.
By inductive assumption, there is a 2-coloring $\varphi'$ of $H'$, which requires at least $\lfloor (n-\ell)/2 \rfloor$ recoloring operations to be made proper. Partition the set $U$ into $2$ equal subsets and extend $\varphi'$ to a coloring $\varphi$ of $H$, by assigning a distinct color to each of these subsets. Note that the number of recoloring operations required to make $\varphi$ proper is at least $\lfloor (n-\ell)/2 \rfloor + \ell/2 = \lfloor n/2 \rfloor$.

If $\ell$ is odd, let $H'$ be the graph obtained from $H$ by removing vertices from $U \cup \{v\}$. By inductive assumption there is a coloring $\varphi'$ of $H'$, which requires at least $\lfloor (n-\ell-1)/2 \rfloor$ recoloring operations. Moreover, by Proposition~\ref{prop:star-lower} there is a coloring $\varphi''$ of the star $H[U \cup \{v\}]$, which requires at least $\lceil \ell/2 \rceil = (\ell+1)/2$ recoloring operations. Let $\varphi$ be the coloring of $H$ obtained as a union of $\varphi'$ and $\varphi''$ (note that the domains of these colorings form a partition of the vertex set of $H$). Thus, the minimum number of recoloring operations required to fix $\varphi$ is at least $\lfloor (n-\ell-1)/2 \rfloor + (\ell+1)/2 = \lfloor n/2 \rfloor$.

Recall from Corollary \ref{cor:rischi} that $\fixnumber{H} = \fixnumb{}{2}{H} \geq \fixnumb{\varphi}{2}{H} \geq \lfloor n/2 \rfloor$.
The claim follows from a straightforward observation that $\fixnumber{G} \geq \fixnumber{H}$.
\end{proof}

Combining Corollary \ref{cor:bipartite-upper} and Proposition \ref{prop:star-lower}, we obtain the following:
\begin{cor}
For a connected bipartite graph $G$ it holds that $\fixnumber{G} = \lfloor n/2 \rfloor$.
\end{cor}

We believe that it would be interesting to investigate $\fixnumber{G}$ for non-bipartite $G$. In particular, what are graph classes, for which the bound from Theorem \ref{thm-fixupper} is not tight? How does $\fixnumber{G}$ depend on structural parameters of $G$?  Given a graph $G$, how to construct the ,,worst possible'' initial coloring?
Finally, it might be interesting to study the computational problem of determining $\fixnumber{G}$ for a given graph $G$.

\vskip 10pt
\noindent{\textbf{Acknowledgement.}} The authors are sincerely grateful to Dieter Kratsch for valuable discussion on the topic. We are also grateful to anonymous reviewers for their numerous comments, especially for suggesting the bound in Proposition~\ref{prop:lower} and its consequences.

\end{document}